\documentclass[runningheads]{llncs}

\usepackage{epsfig}
\usepackage{latexsym}
\usepackage{enumerate}
\usepackage{url}
\usepackage{float}
\usepackage[hypertexnames=false,hyperfootnotes=false]{hyperref}
\usepackage{texnansi}
\usepackage{color}
\usepackage{tikz}
\usepackage[margin=10pt,font=small,labelfont=bf]{caption}
\usepackage[subrefformat=parens,labelformat=parens]{subcaption}
\usepackage{afterpage}
\usepackage{enumitem}
\usepackage[boxed]{algorithm}
\usepackage{algpseudocode}
\usepackage[normalem]{ulem}
\usepackage{lmodern}
\usepackage{booktabs}
\usepackage{sectsty}
\usepackage{ifthen}
\usepackage{amsmath}
\usepackage{amssymb}
\usepackage{amsfonts}
\usepackage{thmtools}
\usepackage{thm-restate}
\usepackage{mathtools}
\usepackage{xspace}
\usepackage{titling}
\usepackage{natbib}
\usepackage{xfrac}
\usepackage{multirow}
\usepackage{bigdelim}
\usepackage{bm}
\usepackage{cleveref}
\usepackage[textsize=tiny]{todonotes}
\usepackage{booktabs} 
\usepackage{array} 
\newcolumntype{P}[1]{>{\centering\arraybackslash}p{#1}} 
\newcommand{\field}[1]{\ensuremath{\mathbb{#1}}}
\newcommand{\R}{\ensuremath{\field{R}}} 
\newcommand{\defeq}{\ensuremath{\triangleq}}
\newcommand{\subjectto}{\text{\rm subject to}} 

\DeclareMathOperator{\diag}{diag}

\DeclareMathOperator*{\argmin}{\mathrm{argmin}}
\DeclareMathOperator*{\argmax}{\mathrm{argmax}}
\newcommand{\minimize}{\ensuremath{\mathop{\mathrm{minimize}}\limits}}

\renewcommand{\thmcontinues}[1]{\hyperref[#1]{continued}}
\usetikzlibrary{arrows,patterns,plotmarks,pgfplots.groupplots}
\tikzstyle{every picture} += [>=stealth]
\tikzset{axis/.style={semithick, line join=miter}}
\allsectionsfont{\sffamily}
\makeatletter
\def\@seccntformat#1{\csname the#1\endcsname.\quad}
\makeatother
\renewcommand{\figurename}{\normalfont\sffamily\bfseries Figure}
\renewcommand{\tablename}{\normalfont\sffamily\bfseries Table}
\renewcommand{\abstractname}{\normalfont\sffamily\bfseries Abstract}
\floatname{algorithm}{\normalfont\sffamily\bfseries Algorithm}
\floatstyle{ruled}
\newcommand{\emailhref}[1]{\href{mailto:#1}{\tt #1}} 
\provideboolean{fastcompile}
\newcommand{\hidefastcompile}[1]{\ifthenelse{\boolean{fastcompile}}{}{#1}}
\usepackage{pgfplots}
\usepackage{pgfplotstable}
\usetikzlibrary{calc}
\usepackage{mathtools}
\definecolor{orange}{rgb}{0.85,0.33,0.13} 
\definecolor{green}{rgb}{0.13,0.85,0.33}
\definecolor{purple}{rgb}{0.33,0.13,0.85}
\definecolor{lime}{rgb}{0.65,0.85,0.13}
\definecolor{blue}{rgb}{0.13,0.65,0.85}
\pgfplotscreateplotcyclelist{tricolor}{%
  orange,every mark/.append style={fill=orange!80!black},mark=*\\%
  green,every mark/.append style={fill=green!80!black},mark=square*\\%
  purple,every mark/.append style={fill=purple!80!black},mark=otimes*\\%
  black,mark=star\\%
  orange,every mark/.append style={fill=orange!80!black},mark=diamond*\\%
  green,densely dashed,every mark/.append style={solid,fill=green!80!black},mark=*\\%
  purple,densely dashed,every mark/.append style={solid,fill=purple!80!black},mark=square*\\%
  black,densely dashed,every mark/.append style={solid,fill=gray},mark=otimes*\\%
  orange,densely dashed,mark=star,every mark/.append style=solid\\%
  green,densely dashed,every mark/.append style={solid,fill=green!80!black},mark=diamond*\\%
}
\pgfplotsset{colormap={tricolormap}{color=(orange) color=(green) color=(purple)},
  colormap={quadcolormap}{color=(orange) color=(lime) color=(blue) color=(purple)}}
\pgfplotstableset{%
  font=\small,
  every head row/.style={before row=\toprule[1pt], after row=\midrule},
  every last row/.style={after row=\bottomrule[1pt]}}
\pgfplotsset{compat=1.15}

\usepackage{setspace}
\usepackage[margin=1in]{geometry}

\provideboolean{submissionversion}
\setboolean{submissionversion}{true}
\provideboolean{blindversion}

\ifthenelse{\boolean{submissionversion}}{
  \renewcommand{\todo}[2][1]{}
  
  \newcommand{\deledit}[1]{}
}{
  
  \newcommand{\deledit}[1]{{\color{orange} \sout{#1}}}
}

\spnewtheorem{assumption}{Assumption}{\bfseries}{\itshape}
\crefname{assumption}{assumption}{assumptions}

\newcommand{\aim}{\ensuremath{\mathsf{aim}}}

\title{\bf\sffamily Optimal Dynamic Fees for Blockchain Resources}
\ifthenelse{\boolean{blindversion}}{
  \author{Anonymous Authors}
}{
  \author{
    Davide Crapis \\
    Robust Incentives Group\\
    Ethereum Foundation \\
    \emailhref{davide@ethereum.org} \\
    \and
    Ciamac C. Moallemi \\
    Graduate School of Business \\
    Columbia University \\
    \emailhref{ciamac@gsb.columbia.edu} \\
    \and
    Shouqiao Wang \\
    Graduate School of Business \\
    Columbia University \\
    \emailhref{shwang27@gsb.columbia.edu} \\
  }
  \date{Initial version: May 11, 2023 \\
    Current version: September 20, 2023
  }
}

\begin{document}
\maketitle

\begin{abstract}
  We develop a general and practical framework to address the problem of the optimal design of dynamic fee mechanisms for multiple blockchain resources. Our framework allows to compute policies that optimally trade-off between adjusting resource prices to handle persistent demand
  shifts versus being robust to local noise in the observed block demand. In the general case with more than one resource, our optimal policies correctly handle cross-effects (complementarity and substitutability) in resource demands. We also show how these cross-effects can be used to inform resource design, \textit{i.e.} combining resources into bundles that have low demand-side cross-effects can yield simpler and more efficient price-update rules. Our framework is also practical, we demonstrate how it can be used to refine or inform the design of heuristic fee update rules such as EIP-1559 or EIP-4844 with two case studies. We then estimate a uni-dimensional version of our model using real market data from the Ethereum blockchain and empirically compare the performance of our optimal policies to EIP-1559.
\end{abstract}


\section{Introduction}


Users of public permissionless blockchains can modify the shared state of the network through \textit{transactions} that are executed by a set of nodes with limited computational resources. To allocate resources among competing transactions most blockchains use \textit{transaction fees}. Initial transaction fee mechanisms in the Bitcoin and Ethereum blockchains relied on
users bidding for transaction inclusion as the main way of pricing congestion. Moreover, all computational resources were bundled into a unique virtual resource (``gas'') with fixed relative prices hardcoded in the protocol. Current R\&D efforts are focused on improving transaction fee markets along two directions: (1) setting a minimum \textit{dynamic base fee} (henceforth also called \textit{price}) that is adjusted by the protocol as function of user demand and (2) \textit{unbundling resources} so that different resources can be individually priced and their relative prices can also efficiently adjust with demand.

In this paper, we propose a new framework for choosing a resource pricing policy that makes significant progress across both directions. We consider the practical problem of a blockchain protocol that has to jointly update the prices of multiple resources at every block. We assume that the type of resources being metered and priced, as well as the block limits and sustainable targets for each resource, are pre-determined. These higher level decisions are the outcome of a design process that has interesting political, economic, and engineering considerations but are outside the current scope of our framework\footnote{We note that some of our results do offer some insights on the resource design problem that we will briefly discuss.}.

Our framework is both general and practical. Or main results characterize theoretically optimal policies in a realistic setting with multiple resources and time-varying demand. Our results can be used in two ways: (i) the policies can be \textit{directly} implemented as we demonstrate, or (ii) insights from our main results can be used to construct and refine heuristics that approximate optimal policies. The latter point is particularly important in the blockchain environment, where, especially at Layer 1, the price computation itself is significantly resource constrained\footnote{Layer 2s, depending on their architecture, can perhaps implement price policies that require more computation and are closer to the optimal ones.}. We designed our framework with the following properties in mind:

\begin{enumerate}
\item[$\bullet$] \textit{Simplicity}: the price update algorithm should be simple enough that it can potentially be implemented without error and tested by multiple clients, the price computation should not consume significant network resources, and the algorithm should be future-proof (it should work with minimal to no upgrades required).
\item[$\bullet$] \textit{Robustness}: the price updates should be robust to sudden and substantial shifts in resource demand as well as fee manipulation attacks (\textit{e.g.}, a full block attack).
\item[$\bullet$] \textit{Optimality}: the performance of the algorithm should satisfy some criterion of optimality.
\end{enumerate}

Our main theoretical contributions are the design of a flexible modeling framework that satisfies most of these properties and the characterization of optimal dynamic pricing policies with multiple resources and time-varying demand under a diverse set of assumptions. In particular:
\begin{enumerate}
\item Our framework allows to specify an arbitrary number of resources to be priced, it models the aggregate resource demand dynamics as a system of linear equations and the realized block resource demand as a system of linear functions of resource prices. It uses as objective function the quadratic deviation of block usage from sustainable target at each block plus a regularization term for controlling price fluctuations. In this way, we cast the blockchain resource pricing problem into the linear-quadratic framework that has been studied in the optimal control literature and widely used in practical applications, from rocket control to congestion pricing in ride-sharing networks.
\item We derive the optimal pricing policy for the most generic formulation with \textit{n} resources and price regularization. The price update optimally handles the trade-off between adjusting prices in response of persistent resource demand shifts versus being robust to local noise in the observed block demand. Moreover, it correctly takes into account the cross-effects in the joint price updates -- \textit{i.e.}, accounting that the demand for resource A will be suppressed (increased) as a result of a price increase of resource B when these resources are complementary (substitutable), and vice-versa.
\item We present a novel decomposition of the resource pricing problem into a set of independent problems, one for each resource. \textit{We use the structure of the demand function to define what are effectively bundles of resources}, whose optimal prices can be computed independently. We call these bundles eigenresources and their prices eigenprices. Our result shows that when this decomposition is possible, optimal prices can be constructed from simple independent update rules. Moreover, the result offers actionable insights to protocol researchers that are thinking about market design for multiple resources.
\end{enumerate}

We also take the optimal policies to real market data. We collect historical data on fourteen days and close to one hundred thousand blocks of the Ethereum blockchain. We estimate our model parameters using the EM algorithm and then compute the policies that we characterized in our theoretical results. We introduce an evaluation framework where we compare our policies to EIP-1559 across a set of performance metrics. The policies perform well across the set of metrics, with the regularized policy beating the benchmark on all metrics. We have reason to believe that our framework can perform even better in multi-dimensional settings with significant cross-effects between resources and we plan to extend this empirical analysis to two-dimensions as real market data on EIP-4844 becomes available.

On top of providing optimal pricing policies and algorithms to compute them, we also show how they can be used to inspire or refine simpler heuristics that are used in practice. We present two case studies for EIP-1559 and EIP-4844. In the case of EIP-1559, we derive a heuristic based on constant price sensitivity and show how the \textit{price update} can be efficiently computed as a function of the resource \textit{demand estimate}. The former is already computed by most blockchains that use an EIP-1559 type formula to update base fees at every block, where usually the latest block demand is used as a na$\ddot{\i}$ve single point estimate of demand. Our heuristic only requires to maintain one additional variable, the demand estimate, and both resource demands and prices can be updated after each block with a simple linear update that only uses local information.


A final thing to clarify is that the linear-quadratic optimization techniques we use to derive our optimal policy have been widely studied and used. Our contribution is to cast the blockchain resource pricing problem as a variant of the linear-quadratic framework, and define models for the resource demand dynamics and demand system that fit in the framework and yet can accommodate for many practical scenarios. We also note that the relation between EIP-1559 type rules and optimization of dynamical systems has been recently studied in the literature as we discuss in the next section. The novelty of our framework is to use explicit models for the time-varying demand dynamics, something that has not been considered before. This allows us to derive dynamic pricing policies in a changing environment that depend on primitives of the demand system such as price sensitivity/cross-sensitivity and variance.


\subsection{Literature Review}

Blockchain transaction fee mechanisms have recently undergone major improvements and are the
subject of current R\&D in the Ethereum ecosystem. The first major set of improvements was
originally researched by \cite{buterin2018blockchain} in a seminal paper that proposed a series of
ground-breaking ideas, among which justifying the use of a dynamic \textit{base fee} that is
updated by the protocol to avoid congestion. This and related ideas such as burning base fees to
align incentives were implemented in the Ethereum protocol\footnote{It has also been forked and
  adopted by many other protocols such as Filecoin, NEAR, and Ethereum L2 chains.} via
EIP-1559. \cite{tim2020eip1559} provides a thorough game-theoretic analysis of this transaction
fee mechanism, and also discusses in depth the base fee update rule focusing on desiderata,
potential attack vectors, and alternative designs. \cite{reijsbergen2022transaction} provides a
performance analysis using on-chain data and explores alternative designs.

The multi-dimensional transaction fee mechanism is a further upgrade that has recently been proposed by \cite{buterin2022multidimensional} and has led to a new wave of interest on this topic. On the applied R\&D side, EIP-4844 implements the first bi-dimensional fee mechanism for blockchains with the creation of the \textit{Ethereum datagas} resource and market which is one of the keys to Ethereum scalability, see \cite{crapis2023eip4844} for an analysis and simulation. 

A few recent papers consider the question of designing optimal fee update rules. Perhaps the
closest to our work is that of \cite{diamandis2022dynamic} which focuses on the multi-dimensional
case and derives optimal prices as the solution of a welfare maximization problem and a general
class of convex functions for the network cost. While the above focuses on a \emph{static} problem
with full information, in our present work we model both demand dynamics and the uncertainty of
the protocol, that has to make optimal decisions based on noisy signals of time-varying demand. \cite{leonardos2022optimality} study the optimality of the EIP-1559 price update rule and alternative formulations in the uni-dimensional case through the lens of dynamical systems analysis. \cite{ferreira2021dynamic} study the stability and welfare optimality of a dynamic posted-price mechanism that conditions prices on past block utilization and observable bids, they focus on steady state analysis of a system with one resource and a probabilistic demand distribution that does not change over time.


\section{Model}\label{sec:model}

In what follows we describe the main components of our framework. We consider a setting with a realistic demand process for the general \textit{n}-dimensional resource pricing problem. We model aggregate resource demand at every block, in discrete time indexed by $k$, and describe the protocol decision problem.

\medskip
\noindent\textbf{\sffamily Demand dynamics.} The aggregate resource demand $d_k$ of all user
transactions that can potentially be included in block $k$ is governed by mean-reverting process
dynamics with
Gaussian noise given by
\begin{equation}
    d_{k+1} = (I-A^{d})\mu^{d} + A^{d} d_{k} + \epsilon_k^{d}.
\label{d_dynamic}
\end{equation}
Here, $d_k\in\mathbb{R}^n$ is the resource demand vector in block $k$,
$\epsilon_k^{d} \sim \mathcal{N}(0, W^{d})$ is i.i.d.\ Gaussian noise with covariance matrix
$W^{d}\in\mathbb{R}^{n \times n}$, and $\mu^{d}\in\mathbb{R}^n$ and
$A^{d}\in\mathbb{R}^{n \times n}$ are parameters of the demand process.  This process captures the
essential properties of blockchain resource demand. The long-term mean base load of resources is
given by $\mu^{d}$, but demand can spike and deviate to high load and revert to the base load, the
matrix $A^{d}$ regulates the permanence of demand spikes for different resources.\footnote{The
  more general case in which $\mu^{d}$ and $A^{d}$ are time-varying also fit in our framework, we
  consider the constant case for a clearer exposition.}

\medskip
\noindent\textbf{\sffamily Observed demand.} The protocol observes the aggregate resource demand
$y_k \in \mathbb{R}^n$ of all transactions that are included into a valid block. We assume that
$y_k$ is given by the demand model
\begin{equation}\label{eq:obs-demand}
  y_{k} = d_{k} + B_{k} p_{k} + \epsilon_{k}^{y},
\end{equation}
where $p_{k}\in\mathbb{R}^n$ are resource prices (determined by the protocol),
$B_{k}\in\mathbb{R}^{n \times n}$ are price-sensitivity parameters, and $\epsilon_{k}^{y} \sim \mathcal{N}(0, W^{y})$
is i.i.d.\ Gaussian noise with covariance matrix $W^{y}\in\mathbb{R}^{n \times n}$. That is, the
protocol only observes realized demand for all resources, which we model as a linear decreasing
function of resource prices. We also allow for cross-sensitivity, \textit{i.e.}, if the matrix entry
$B_{k,12}$ is non-zero, the price of
resource~$2$ can have a positive (respectively, negative) impact on demand of resource~$1$, if the
the resources are complements (respectively, substitutes).

\medskip
\noindent\textbf{\sffamily Price sensitivity.} The price-sensitivity matrix $B_{k}$ determining
the observed demand $y_{k}$ via  \eqref{eq:obs-demand} evolves according to
\begin{equation}
    \mathbf{vec}(B_{k+1}) = (I-A^{B})\mu^{B}+A^{B}\mathbf{vec}({B_k})+\epsilon_k^{B},
\label{B_dynamic}
\end{equation}
where $\mathbf{vec}(\cdot)$ is a function that reshapes a matrix into a column vector by stacking
its columns, $A^{B}\in\mathbb{R}^{n^2 \times n^2}$, $\mu^{B}\in\mathbb{R}^{n^2}$, and
$\epsilon_k^{B} \sim \mathcal{N}(0, W^{B})$ is i.i.d.\ Gaussian noise with covariance matrix
$W^{B}\in\mathbb{R}^{n^2 \times n^2}$. It follows a mean-reverting process, where the long-term
mean is given by the vector $\mu^{B}$, and the parameters $A^{B}$ quantifies how quickly
$\mathbf{vec}(B_k)$ reverts to the base $\mu^{B}$.

\medskip
\noindent\textbf{\sffamily Price update.} The resource prices for block $k+1$ are decided after observing demand at block $k$. Characterizing optimal prices is the subject of the following section.

\medskip
\noindent\textbf{\sffamily Objective.} Our objective function is a weighted sum of two loss functions. First, we aim to find a pricing policy that minimizes resource demand deviation from the protocol block resource target $t\in\mathbb{R}^n$. We choose the long-term average of infinite-horizon aggregation of a quadratic loss function,
$$\lim_{K \to \infty} \frac{1}{K} \sum_{k=1}^{K} \| y_{k} - t \|_{2}^{2}.$$
Second, we aim to find a price update policy that minimizes price fluctuations over time. We choose the long-term average of infinite-horizon quadratic variation of the price process as our second loss function,
$$\lim_{K \to \infty} \frac{1}{K} \sum_{k=1}^{K} \| p_{k} - p_{k-1} \|_{2}^{2}.$$
We set $\lambda \geq 0$ as the weight for these two loss functions. We can get our objective function to be
$$\lim_{K \to \infty} \frac{1}{K} \sum_{k=1}^{K} \left( \| y_{k} - t \|_{2}^{2} + \lambda \| p_{k} - p_{k-1} \|_{2}^{2} \right).$$
The total loss assigns increasingly high penalty to higher deviation from target demand and more
significant fluctuations in price process.\footnote{We can also specify weights and penalize
  deviation of resources differently, but we assume uniform weights or alternatively that resource
  measurements are already properly re-scaled to focus on the fundamental trade-offs between
  resources.} Note that this is a problem that involves uncertainty, in the next section we show
how we can apply stochastic control techniques to find an optimal policy that minimizes the
expected objective function.


\section{Optimal Policy}
\label{sec-optimal-policy}

In this section, we delve into a comprehensive analysis of the optimal pricing policy. We first
describe the information structure and belief update techniques that are required to solve the
stochastic control problem. We then present two approaches to solve the problem, under different
assumptions. 

\medskip
\noindent\textbf{\sffamily Information structure.} We first define the information structure at
every block $k$. The \textit{hidden state} is given by the vector
$$x_k^\intercal \defeq \left[ d_k^\intercal \  \mathbf{vec}(B_k)^\intercal \right].$$
It evolves according to
$$x_{k+1}=(\mathrm{I}-A^{x})\mu^{x}+A^{x}x_k+\epsilon_k^{x}, \quad\quad \epsilon_k^{x} \sim \mathcal{N}(0, W^{x}),$$
where the parameters $(A^{x},\mu^{x},W^{x})$ can be determined from the parameters
$(A^d,\mu^d,W^d,A^B,\mu^B,W^B)$ based on equations \eqref{d_dynamic} and \eqref{B_dynamic}.  Here,
we allow the dependencies between $\epsilon_k^{d}$ and $\epsilon_k^{B}$, so that the noise term
$\epsilon_k^{x}$ is i.i.d.\ Gaussian noise with covariance matrix $W^{x}$.  The
\textit{observation equation} is
\begin{equation}\label{eq:obs}
  y_k=d_k+B_k p_k+\epsilon_k^{y}, \quad\quad \epsilon_k^{y} \sim \mathcal{N}(0, W^{y}),
\end{equation}
where $\epsilon_k^{y}$ is i.i.d.\ and independent with $\epsilon_k^{x}$.
The \textit{information available} to the controller is
$$I_k \defeq \sigma(y_0, y_1, ..., y_k, p_0, p_1,..., p_k).$$
We require that policies be adapted to the filtration $\{ I_k\}$.

\medskip
\noindent\textbf{\sffamily Belief updates.} We use the Kalman filter to derive the formulas for
the Bayesian update of beliefs on the state $x_k$ and on future observation $y_k$. The observation
equation \eqref{eq:obs} can also be expressed as
$$y_k = C_k x_k + \epsilon_k^{y},$$
for an appropriate choice of the matrix $C_k$ (depending on the price vector $p_k$).
Using a conjugate Gaussian prior, denote the prior at $k$ by
$$x_{k-1} \sim \mathcal{N}(\hat x_{k-1}, \hat \Sigma_{k-1}).$$
The \textit{predictive distribution} of $x_k$ is Gaussian $\mathcal{N}(a_k,S_k)$ with parameters
\[
  \begin{split}
    a_k &= \mathbf{E}(x_k|I_{k-1})=(\mathrm{I}-A^{x}) \mu^{x} + A^{x} \hat x_{k-1},\\
    S_k &= \mathbf{Var}(x_k|I_{k-1})=A^{x} \hat \Sigma_{k-1} (A^{x})^\intercal + W^{x}.
  \end{split}
\]
The \textit{predictive distribution} of $y_k$ is Gaussian $\mathcal{N}(f_k,F_k)$ with parameters
\[
  \begin{split}
    f_k&=\mathbf{E}(y_k|I_{k-1}) = C_k a_k,\\
    F_k&=\mathbf{Var}(y_k|I_{k-1}) = C_k S_k C_k^\intercal + W^{y}.
  \end{split}
\]
The \textit{posterior distribution} of $x_k$ is Gaussian $\mathcal{N}(\hat x_k,\hat\Sigma_k)$ with parameters
\[
  \begin{split}
    \hat x_k &= \mathbf{E}(x_k|I_{k})=(\mathrm{I}-K_k C_k)a_k+K_k y_k,\\
    \hat \Sigma_k &= \mathbf{Var}(x_k|I_{k})=(\mathrm{I}-K_k C_k)S_k.
  \end{split}
\]
where $K_k\defeq S_k C_k^\intercal(C_k S_k C_k^\intercal+W^{y})^{-1}$ is the Kalman gain. This is the standard Kalman filter update which computes the belief on the current hidden state given the current information.

\medskip
\noindent\textbf{\sffamily Optimization problem.} In the following results, we will characterize
control policies that seek to minimize the expected long-term average cost, solving the problem
\begin{equation}
  \begin{array}{ll}
    \minimize_{p} & \displaystyle
                    J \defeq \lim_{K \to \infty} \mathbf{E} \left\{
                     \frac{1}{K} \sum_{k=1}^{K} \left( \| y_{k} - t \|_{2}^{2} +\lambda \| p_{k} - p_{k-1} \|_{2}^{2} \right) \right\} \\
    \subjectto & \text{\rm $p_{k+1}$ is $I_{k}$-adapted}.
  \end{array}
  \label{problem}
\end{equation}

\subsection{Target Loss Objective}
Our first result focuses on the case where we only aim to find a price update policy that
minimizes resource demand deviation from the target $t$. In particular, we do not add any penalty
for fluctuations in the price process and we make the following assumption.

\medskip
\begin{assumption}
\label{assumption_lambda_0}
    $\lambda=0$.
\end{assumption}

\medskip
\noindent The following theorem characterizes the optimal policy under Assumption~\ref{assumption_lambda_0}.
\begin{theorem}
  Consider Problem \eqref{problem} and suppose Assumption \ref{assumption_lambda_0} holds. Then, the optimal policy updates prices according to the rule
  $$p_{k+1}^* = \left[\mathbf{E}\left(B_{k+1}^\intercal B_{k+1}|I_k\right)\right]^{-1}\mathbf{E}\left(B_{k+1}^\intercal (t-d_{k+1})|I_k\right).$$
\label{thm-lambda-0}
\end{theorem}
\begin{proof}
The proof is provided in Appendix~\ref{app:proofs}.
\end{proof}

Under Assumption \ref{assumption_lambda_0}, our focus is directed solely towards achieving the primary objective: making the observed, or realized, demand align as closely as possible with our predefined target $t$. This approach can be viewed as a pursuit of optimal demand estimation. We use Kalman filter to facilitate joint updates across our potential demand $d_k$ and the price sensitivity matrix $B_k$. This dynamic filtering technique enables us to incorporate new information into our estimation process. The optimal vector of resource prices for block $k$ is the one that matches the expected demand of each resource to the respective target. Note that in general, the optimal policy necessarily takes into account the expected demand, price sensitivity, and target of resource $j$ when updating the price of resource $i$. Whenever demand-side cross-effects are present the optimal prices need to be set jointly.

\subsection{General Setting: Model Predictive Control}

We now turn to the general problem, in which we have a penalty for fluctuations in the price
process, with weight $\lambda > 0$. This is important because large pricing fluctuations between
adjacent blocks directly degrade user experience. Thus, controlling the magnitude of price update
is a requirement for every policy that can be adopted in practice and it was one of the main
design goals of EIP-1559.

Our framework allows us to flexibly control price fluctuations with the additional regularization
term. However, this increases the computational complexity of the optimal policy. One particular
challenge is that the price-sensitivity matrix $B_k$ is stochastic. Since this multiplies the
decision variable $p_k$, our setting does not fall into the standard framework of
linear-quadratic-Gaussian (LQG) control.

This leads us to to the following heuristic procedure: first, we observe that if the evolution of
price-sensitivity matrix $B_k$ is deterministic, then the problem can be optimally solved using
standard LQG control methods. Then, in the stochastic setting, we repeatedly apply the optimal
deterministic policy, always using the then current estimate for $B_k$, but pretending (for the
purpose of control) that the future evolution of $B_k$ is deterministic --- this is the idea of
model predictive control.

\medskip
\noindent\textbf{\sffamily Optimal policy under deterministic price-sensitivity.}
For the moment, we will make the following assumption:
\begin{assumption}
\label{assumption_deterministic_B}
    The price-sensitivity matrix $B_k$ is deterministic or, in other words, the noise term
    $\epsilon_k^{B} \equiv 0$, for all $k \geq 0$.
\end{assumption}
Under this assumption, the dynamics for $B_k$ are given by
\begin{equation}\label{eq:deterministic_B}
  \mathbf{vec}(B_{k+1}) = (I-A^{B})\mu^{B}+A^{B}\mathbf{vec}({B_k}).
\end{equation}
Using LQG control methods, we are able to derive an optimal policy explicitly,
and the following theorem characterizes the optimal policy:

\begin{theorem}
  Consider problem \eqref{problem} and suppose \Cref{assumption_deterministic_B} holds. Then, the optimal policy updates prices according to the rule
  \begin{equation}\label{eq:p-opt}
    p_{k+1}^* =\left( I-\frac{Q_k}{\lambda} \right) p_k + \frac{Q_k}{\lambda} \aim_k.
  \end{equation}
  Here, $Q_k \in \R^{n\times n}$ is a positive definite matrix capturing the quadratic coefficient
  of the price vector $p_k$ the cost-to-go function (we provide an explicit derivation for this as
  the solution of a Ricatti equation in the Appendix~\ref{app:proofs}).  $\aim_{k}$ is the ``aim'' price defined by
  \begin{equation}\label{eq:aim}
    \aim_{k} = \sum_{s=k+1}^\infty (I-Z_{s})Z_{s-1}Z_{s-2}\cdots Z_{k+1}\Bar{p}_{s},
  \end{equation}
  where $Z_s \defeq \left(Q_s+B_s^\intercal B_s\right)^{-1}Q_s$ is a positive definite matrix, and
  $\Bar{p}_{s}$ is the market clearing price that satisfies
  $$\mathbf{E}(y_s|I_{k},p_s = \Bar{p}_{s})=t.$$
  The sum of the coefficient matrices of the aim price
  satisfies
  $$\sum_{s=k+1}^\infty (I-Z_{s})Z_{s-1}Z_{s-2}\cdots Z_{k+1} = I.$$
\label{thm-lambda-positive}
\end{theorem}
\begin{proof}
The proof is provided in Appendix~\ref{app:proofs}.
\end{proof}

We can understand \Cref{thm-lambda-positive} through an interpretation analogous to that developed
by \citet{garleanu2013dynamic}, who consider different LQG control problem in a quantitative
trading setting.  First, observe that according to \eqref{eq:p-opt}, the optimal price in block
$k+1$ can be expressed as the (matrix-weighted) convex combination of the current price $p_k$ and
a future ``aim'' price $\aim_k$. This captures a trade-off between the two objectives of not
changing prices much (by staying close to the current price) and ensuring that future observed
demand aligns with the target $t$ (by targeting the aim price).

The aim price $\aim_k$, in particular, is a (matrix-weighted) average of future market clearing
prices $\{\Bar{p}_{s}\}$. Each $\Bar{p}_{s}$ is the price for block $s$ that so
that the expected demand in block $s$ matches the target $t$, conditional on information known up
to block $k$. The matrix weights in \eqref{eq:aim} specify exactly how much to weight to put on
different future periods. These weights guide the controller in how much to emphasize targeting
demand in different future blocks, when making the pricing decision for the next block. This
captures the idea that one should set immediate prices in a way that not only works well for the
next block, but also leaves the controller in a good position for subsequent blocks downstream.

\medskip
\noindent\textbf{\sffamily Model predictive control.}
Of course, the original problem \eqref{problem} allows for stochastic and time-varying
price-sensitivity, and this is an important feature in the empirical analysis of
\Cref{sec:empirical}. Therefore, we will adapt the policy of \Cref{thm-lambda-positive} to this
more general setting using the idea of model predictive control (MPC):
\begin{enumerate}
\item After each block $k$, we estimate the current price-sensitivity according to
  $\hat B_k = \mathbf{E}(B_k|I_k)$, using all available information available at that time, under
  the original dynamics \eqref{B_dynamic}. This can be accomplished using a standard Kalman
  filter, as described in the belief updates of \Cref{sec:model}.

\item Given the estimate $\hat B_k$, we solve for the optimal policy under
  \Cref{thm-lambda-positive}, that is, assuming that $B_k=\hat B_k$ and that the future evolution
  of $B_k$ is deterministic as per \eqref{eq:deterministic_B}.
\end{enumerate}

While the MPC policy is no longer optimal for the problem \eqref{problem} under stochastic
price-sensitivity, we shall see in the empirical results of \Cref{sec:empirical} that nevertheless
this procedure is a powerful heuristic.


\section{Resource Design and Separable Eigenresources}

In this section, we present a pricing approach which holds significant economic
implications. Instead of pricing individual resources as explicitly specified, we focus on pricing
the linear combinations of these resources, which we refer to as eigenresources. Our main result
also has implications for protocol designers thinking about how to structure supply-side resources
and create markets for pricing them. The idea here is that, in the policies of
\Cref{sec-optimal-policy}, in general resources must be \emph{jointly} priced. This is because
there could be cross-price elasticities, since some resources can be complements or
substitutes. In this section, we discuss a setting where, through the definition of appropriate
resources, the problem nevertheless separates into different control problems for each resource.

We will derive our result under the assumption that the price-sensitivity parameter $B_{k}$ can be
decomposed into principal components consistently over time:
\begin{assumption}
\label{assumtion_eigenresources}
There exist fixed, orthogonal matrices $U,V\in\mathbb{R}^{n \times n}$ such that, for all $k$,
$$B_{k} = U \diag(\delta_{k}) V^\intercal.$$
Here, $\delta_{k}\in\mathbb{R}^{n}$ is the sensitivity vector for eigenresources governed by a mean-reverting process
\begin{equation}
  \delta_{k+1} = (I-A^{\delta}) \mu^{\delta} + A^{\delta} \delta_{k} + \epsilon_{k}^{\delta}.
  \label{delta_dynamic}
\end{equation}
where $A^{\delta} \in\mathbb{R}^{n \times n}$, $\mu^{\delta}\in\mathbb{R}^{n}$, and $\epsilon_{k}^{\delta} \sim \mathcal{N}(0, W^{\delta})$ represents the i.i.d.\ Gaussian noise with covariance matrix $W^{\delta}\in\mathbb{R}^{n \times n}$.
\end{assumption}

Given the matrices $(U,V)$, we can change coordinates as follows:

\medskip
\noindent\textbf{\sffamily Eigendemand.} The eigendemand of the eigenresources $\Tilde{y}_k \in \mathbb{R}^n$ that are included into a valid block is given by
$$\Tilde{y}_{k} = U^{-1} y_k.$$
Here, the columns of $U^{-1}$ define linear combinations of the original resources. These new
eigenresources are what will be prices.

\medskip
\noindent\textbf{\sffamily Eigenprice.} The eigenprice of eigenresources $\Tilde{p}_k \in \mathbb{R}^n$ is given by
$$\Tilde{p}_{k} = V^\intercal p_k.$$

\medskip
\noindent\textbf{\sffamily Eigentarget.} The eigentarget of eigenresources $\Tilde{t} \in \mathbb{R}^n$ is given by
$$\Tilde{t} = U^{-1} t.$$
The orthogonal property of matrix $U$ ensures the cost for the demand deviation from target $t$ equals to the cost for the eigendemand deviation from eigentarget $\Tilde{t}$, i.e.
$\| y_{k} - t \|_{2}^{2} = \| \Tilde{y}_{k} - \Tilde{t} \|_{2}^{2},$
which means it is equivalent for us to minimize the eigendemand deviation from eigentarget $\Tilde{t}$.

\medskip
\noindent\textbf{\sffamily Eigenprice sensitivity.} We can get the eigendemand $\Tilde{y}_k$ satisfies
$$\Tilde{y}_k = \Tilde{d}_k + \diag(\delta_k) \Tilde{p}_k +
\Tilde{\epsilon}_k^y,$$
where $\Tilde{d}_k \defeq U^{-1}d_k$ and $\Tilde{\epsilon}_k^y \defeq U^{-1}\epsilon_k^y$ by
$y_{k} = d_{k} + B_{k} p_{k} + \epsilon_k^y.$
We can see that $\delta_k$ is the eigenprice-sensitivity vector indicating the price-sensitivity of each eigenresource. Equation \eqref{delta_dynamic} shows that the long-term mean base eigenprice-sensitivity is given by the vector $\mu^{\delta}$, but it can also fluctuate around the base eigenprice-sensitivity vector, and the parameters $A^{\delta}$ quantifies how quickly the eigenprice-sensitivity vector reverts to the base eigenprice-sensitivity vector.

\medskip
\noindent\textbf{\sffamily Optimal policy.} The following theorem characterizes the optimal policy under Assumption \ref{assumption_lambda_0} and Assumption \ref{assumtion_eigenresources}.

\begin{theorem}\label{thm:eigen}
  Consider the objective of minimizing the expected long-term average cost,
  \[
    \begin{array}{ll}
      \minimize_{p} & \displaystyle
                      J \defeq \lim_{K \to \infty} \mathbf{E} \left\{
                       \frac{1}{K} \sum_{k=1}^{K} \| y_{k} - t \|_{2}^{2} \right\} \\
      \subjectto & \text{\rm $p_k$ is $I_{k-1}$-adapted}.
    \end{array}
  \]
  Then, the optimal policy updates prices according to the rule
  $$ p_{k+1}^* \defeq V \Tilde{p}_{k+1}^*,$$
  where the $i$th element of eigenprice $\Tilde{p}_{k+1}$ satisfies
  $$\Tilde{p}_{k+1,i}^{*} = \frac{\mathbf{E} \left( \delta_{k+1,i} (\Tilde{t}_i-\Tilde{d}_{k+1,i}) |I_k \right)}{\mathbf{E} \left( \delta_{k+1,i}^2 | I_k \right)},$$
  $\delta_{k+1,i},\Tilde{t}_i,\Tilde{d}_{k+1,i}$ are the $i$th element of $\delta_{k+1},\Tilde{t},\Tilde{d}_{k+1}$ respectively.
\end{theorem}
\begin{proof}
  The proof is provided in Appendix~\ref{app:proofs}.
\end{proof}

The optimal price in block $k+1$ is the product of orthogonal matrix
$V$ and the eigenprice
$\Tilde{p}_{k+1}$. The explicit expression of each eigenprice $\Tilde{p}_{k+1,i}^*$ shows that
optimal prices for eigenresources depend only on the its own expected demand, price-sensitivity,
and target. This implies that when one can leverage this decomposition, one can devise pricing
strategies that are simple yet optimal. In particular, they do not require the joint determination
of all prices that can be more computationally intensive.

Note that \Cref{thm:eigen} is analogous to \Cref{thm-lambda-0}, in that they both handle the
$\lambda=0$ case (\Cref{assumption_lambda_0}). There is a similar analogous result for the
deterministic price-sensitivity case of \Cref{thm-lambda-positive}: in that case, the optimal
policy will similarly separate across eigenresources. We omit an explicit statement and proof of
this result since it is relatively obvious given the developments thus far.

Beyond pricing, these results
also holds implications for resource market design. In particular,
if the protocol designer can define virtual resources that have minimal demand-side interaction
between them, then simple heuristics that implement independent price updates between resources
will have a smaller efficiency loss versus the optimal pricing policy. We further expand on this
insight in \Cref{sec-implications} when we discuss EIP-4844.



\section{Implications}
\label{sec-implications}

In this section we study particular instances of our framework (\textit{e.g.}, uni-dimensional or
bi-dimensional resources) and develop practical implications by comparing optimal policies to
heuristics that are used in practice, such as the EIP-1559 and the EIP-4844 price update
rules. For a more direct comparison with these rules, in this section we consider the optimal
policy characterized in \Cref{thm-lambda-0} when $\lambda=0$, and we focus on the special case in
which the price sensitivity $B_k$ is deterministic and constant.

\subsection{Uni-dimensional Resource and Ethereum Gas}

Suppose there is only one resource, for example a virtual resource such as \textit{gas} in
Ethereum. Define $\beta\defeq\partial y_k/\partial p_k<0$ to be the (deterministic, constant) price sensitivity of demand in this one dimensional model, substituting $p_{k+1}=p_k+u_k$ in the optimal pricing policy of Theorem 1 yields the following equation for the optimal price update
\begin{equation}\label{unidim-policy}
    u_k = - \beta^{-1} (a_{k+1}+\beta p_{k} - t),
\end{equation}
where the predictive estimate of demand can be expressed as
$$
a_{k+1} = (1-\alpha) \mu^d + \alpha \left(\frac{\sigma_{\epsilon^y}^2}{\alpha^2\hat\sigma^2_{k}+\sigma_{\epsilon^d}^2+\sigma_{\epsilon^y}^2}a_{k} + \frac{\alpha^2\hat\sigma^2_{k}+\sigma_{\epsilon^d}^2}{\alpha^2\hat\sigma^2_{k}+\sigma_{\epsilon^d}^2+\sigma_{\epsilon^y}^2}(y_k-\beta p_k) \right).
$$
Here, $\alpha$ is the uni-dimensional entry of the matrix $A^d$ and the variances
$\hat\sigma^2_{k},\sigma_{\epsilon^d}^2,\sigma_{\epsilon^y}^2$ are the uni-dimensional entries of the matrices
$\hat{\Sigma}_k, W^d, W^y$, respectively. The \textit{optimal price update} has two main properties:
\begin{itemize}
    \item[\textit{(i)}] \textit{Maintains the optimal balance between responsiveness to change and robustness to noise.} It does this by using the signal-to-noise ratio $\sigma_{\epsilon^y}^2/(\alpha^2\hat\sigma^2_{k}+\sigma_{\epsilon^d}^2)$ to optimally weigh in new information -- the higher the idiosyncratic variance $\sigma_{\epsilon^y}^2$ of block size observation, \textit{vis-\`a-vis} a more persistent state change, the less responsive the optimal controller will be to local deviation from the current estimate.
    \item[\textit{(ii)}] \textit{Optimally computes the next period price to match predicted and target demand}. First computing the predicted demand at the current price $a_{k+1}+\beta p_{k}$ and then adjusting with the price sensitivity. The \textit{optimal price update is increasing in predicted demand and decreasing in price sensitivity} --- \textit{i.e.}, the more sensitive users are to prices the lower the update required to control demand.
\end{itemize}

\medskip
\noindent
\textbf{\sffamily Optimal Policy vs.\ EIP-1559. }%
We compare the optimal price update policy to the EIP-1559 update rule and distill insights that can guide future design and improvements. Consider the generic EIP-1559 style linear price update rule
\begin{equation}\label{1559-rule}
    p_{k+1} = p_{k} \left(1+ \gamma_k \frac{y_{k} - t}{t} \right),
\end{equation}
where EIP-1559 has $\gamma_k^{\text{eip1559}}\defeq 1/8$. We can apply a few simple transformations to the optimal policy to obtain a comparable form. Note
that the price elasticity of block $k$ demand at the target point $t$ is
$\eta_k = \beta p_k / t$. Solving for $\beta$, substituting into the scaling factor of the
optimal update \eqref{unidim-policy}, and substituting the resulting $u_k^*$ into the full price update
equation $p_{k+1}=p_k+u_k^*$ yields
$$
p_{k+1} = p_{k} \left(1+\frac{1}{|\eta_k|} \frac{a_{k+1}+\beta p_{k} - t}{t} \right).
$$
Thus, the optimal policy is similar to the EIP-1559 rule \eqref{1559-rule}, with two key differences:
\begin{itemize}
    \item[$\bullet$] EIP-1559 uses the current block size $y_k$ as a n\"aive estimate for the next period demand, while the optimal policy uses the predictive estimate;
    \item[$\bullet$] EIP-1559 divides the adjustment by a fixed parameter, set to $8$, instead of the optimal parameter which is equal to the absolute value of the current demand elasticity $|\eta_k|$.
\end{itemize}
Another way of comparing is by noting that the optimal policy is actually equivalent to the EIP-1559 style rule \eqref{1559-rule} with step-size
$$
\gamma_k^{\text{opt}} = \frac{1}{|\eta_k|}\frac{a_{k+1}+\beta p_{k} - t}{y_{k} - t},
$$
which shows how, while EIP-1559 updates prices with a constant
 learning rate of $\gamma_k^{\text{eip1559}}=1/8$, this policy optimally adapts the learning rate to the predicted demand.

 \subsection{Bi-dimensional Resource and Ethereum Datagas}

 Consider now the case of two resources. This is akin to how the Ethereum fee market once EIP-4844
 introduces an additional market for \textit{datagas}. We expand on this shortly, for now consider
 the generic bi-dimensional $n=2$ resource case, the optimal update at step $k$ has the form
$$
\begin{bmatrix}
    u_{k,1}\\
    u_{k,2}
\end{bmatrix}=
\frac{1}{\beta_{12} \beta_{21} - \beta_{11} \beta_{22}}
\begin{bmatrix}
    -\beta_{22} & \beta_{12}\\
    \beta_{21} & - \beta_{11}
\end{bmatrix}
\begin{bmatrix}
    \Delta_{k,1}\\
    \Delta_{k,2}
\end{bmatrix},
$$
where $\Delta_{k,i}=(a_{k+1}+B p_{k} - t)_i$ is the predicted deviation of resource $i$ demand from its target and $\beta_{ij}$ the element in row $i$ and column $j$ of $B$. The update to price 1 can thus be expressed as
$$
u_{k,1}=\frac{\Delta_{k,1}-\beta_{12}\beta_{22}^{-1}\Delta_{k,2}}{\beta_{11}-\beta_{21}\beta_{22}^{-1}\beta_{12}},
$$
which computes the optimal joint resource price update taking into account both resource demand balances and cross-sensitivities. We illustrate how this works by considering three scenarios of increasing sophistication:

\begin{itemize}
    \item[\textit{(i)}] if resource 1 is independent of price 2 ($\beta_{12}=0$) then the optimal update is
      similar to the uni-dimensional case, $u_{k,1} = \Delta_{k,1}/\beta_{11}$;
    \item[\textit{(ii)}] if resource 1 is dependent of price 2 ($\beta_{12}\neq0$) but resource 2 is
      independent of price 1 ($\beta_{21}=0$) then the optimal update is $u_{k,1} = (\Delta_{k,1}-\beta_{12}\beta_{22}^{-1}\Delta_{k,2})/\beta_{11}$, where the second term in the numerator accounts for the joint update of price 2. Suppose $\Delta_{k,2}>0$, then the optimal update of price 1 will be adjusted \textit{downwards if the resources are complement} ($\beta_{12}<0$) and \textit{upwards if the resources are substitute} ($\beta_{12}>0$);
    \item[\textit{(iii)}] if both cross sensitivities are positive then the optimal update also accounts for the \textit{indirect effect} of changing price 1 on resource 1 via its impact on the demand and optimal update of resource 2 (i.e., the pathway $\beta_{21} \rightarrow 1/\beta_{22} \rightarrow \beta_{12}$ in the demand system represented in \Cref{fig:example}). The denominator in the full update formula above is the ``net sensitivity'' of resource 1 to price 1.
\end{itemize}

\begin{figure}[h]
    \centering
    \includegraphics[width=0.35\textwidth]{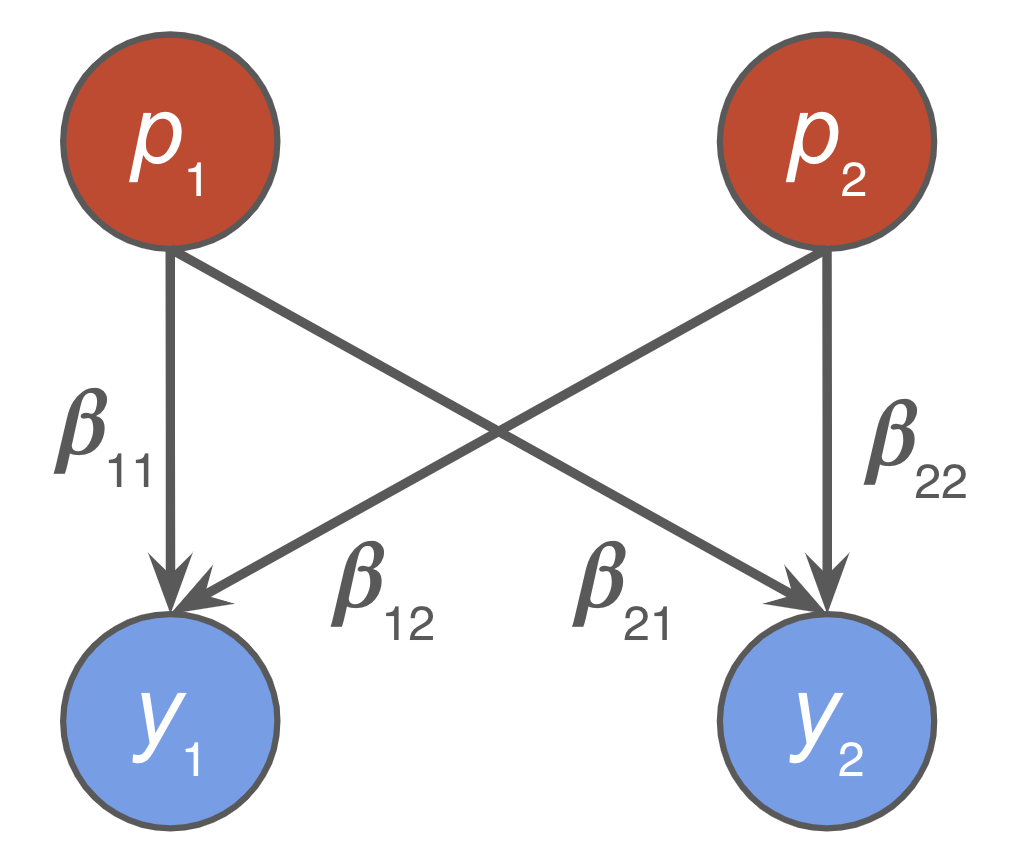}
    \caption{Example of prices-resources relationship in the demand system with cross-sensitivity matrix $B$ (bi-dimensional case).}
    \label{fig:example}
\end{figure}

\medskip
\noindent
\textbf{\sffamily Optimal Policy vs.\ EIP-4844. }%
We now compare the optimal policy to the rule that will be used in Ethereum's fee market after
EIP-4844 is implemented. As mentioned before, this EIP introduces a new resource called
\textit{datagas} on top of the standard \textit{gas}. This new resource is consumed by a
particular type of transaction called \textit{BlobTransaction} which is likely going to be used by
most Layer 2 systems to commit data to Layer 1. The EIP also introduces a new fee market for the
new resource with a price update rule that mimics EIP-1559.\footnote{In reality an exponential form
  is used for \textit{datagas} instead of the linear form used for \textit{gas}. Since the two
  forms are approximately equal we use the linear form for both resources to make the
  presentation clearer.} Thus at every block the prices of both resources are updated as follows
\[
p_{k+1, \text{dat}} = p_{k,\text{dat}} \left(1+ \frac{1}{8} \frac{y_{k,\text{dat}} -
    t_{\text{dat}}}{t_{\text{dat}}} \right), \quad
p_{k+1, \text{gas}} = p_{k,\text{gas}} \left(1+ \frac{1}{8} \frac{y_{k,\text{gas}} -
    t_{\text{gas}}}{t_{\text{gas}}} \right).
\]
One important thing to note is that the updates calculated via these rules are independent from one another, while they are coupled in the optimal update policy. For example, consider the first update rule for datagas above and compare it to the optimal update policy,
\[
  \begin{split}
    p_{k+1, \text{dat}} &=
    p_{k,\text{dat}} +\frac{a_{k+1,\text{dat}}+\beta_{\text{dat},\text{dat}}
                          p_{k,\text{dat}}+\beta_{\text{dat},\text{gas}}
                          p_{k,\text{gas}}-t_{\text{dat}}}{\beta_{\text{dat},\text{dat}}-\beta_{\text{gas},\text{dat}}\beta_{\text{gas},\text{gas}}^{-1}\beta_{\text{dat},\text{gas}}}\\
    &\quad - \beta_{\text{dat},\text{gas}}\beta_{\text{gas},\text{gas}}^{-1}
      \frac{a_{k+1,\text{gas}}+\beta_{\text{gas},\text{dat}}
      p_{k,\text{dat}}+\beta_{\text{gas},\text{gas}}
      p_{k,\text{gas}}-t_{\text{gas}}}{\beta_{\text{dat},\text{dat}}-\beta_{\text{gas},\text{dat}}\beta_{\text{gas},\text{gas}}^{-1}\beta_{\text{dat},\text{gas}}},
  \end{split}
\]
which properly accounts for the effect of an update in gas price on the future demand of datagas and, moreover, it adjusts the datagas price update taking into account the joint update of the gas price. This difference is important because \textit{BlobTransactions} consume both resources and thus the demands are not independent. Failing to take this into account will likely result in cycles of over/under correction of prices. Inefficiencies of this type will be even more prevalent when moving beyond two resources, so we believe that future designs for multi-dimensional resource pricing should have some of the good properties of the optimal pricing policy highlighted here.


\section{Empirical Analysis}\label{sec:empirical}

We now demonstrate how, besides generating theoretical insights and benchmarks, our framework can easily be used in applied settings and calibrated or estimated using onchain data. We focus on a uni-dimensional application, for this we can use historical onchain data for gas
used and the associated base fees that are computed via the EIP-1559 rule that is live on Ethereum
mainnet. In the future, once the data on other resources will be available, for example the data market of EIP-4844 we can extend the analysis to the multi-dimensional case.

\subsection{Empirical Setup}

\medskip
\noindent\textbf{\sffamily Data.} We use historical market data for 99,547 blocks that were added to Ethereum in the span of two weeks (starting on February 24, 2023). The first block is 16,694,514 and the last block is 16,794,061, we rescale the first block number in our dataset to 0 for clarity. We also note that this is a generic 14-day period that does not represent any particular market conditions. For every block we observe the prevailing EIP-1559 base fee and we compute the total gas consumed aggregating over all the transactions included in the block.

\medskip
\noindent\textbf{\sffamily Parameter Estimation.} We estimate the uni-dimensional version of our model using the Expectation-maximization (EM) method on the observed sample. In particular, given the sample observations and the full unobserved state, respectively
$$\mathcal{I}_T\triangleq\left\{(y_{0},p_{0}),(y_{1},p_{1}),\cdots,(y_{T},p_{T})\right\} \quad \text{and} \quad \mathcal{S}_T\triangleq\left\{(d_{0},\beta_{0}),(d_{1},\beta_{1}),\cdots,\allowbreak (d_{T},\beta_{T})\right\},$$
we want an estimator for our model $\theta\triangleq\left(
\mu^d, \mu^{\beta}, \alpha^d, \alpha^{\beta}, \sigma_{\epsilon^d}, \sigma_{\epsilon^{\beta}}, \rho_{\epsilon^d, \epsilon^{\beta}}, \sigma_{\epsilon^y}, d_0, \beta_0, \sigma_{d_0}, \sigma_{\beta_0}, \rho_{d_0, \beta_0}\right)$, where $\alpha^{d}$ and $\alpha^{\beta}$ are the one dimensional entries of the respective matrices $A^{d}$ and $A^{\beta}$, $\sigma$ denotes the standard deviation and $\rho$ the correlation coefficient of the respective noise terms, and the last five parameters are the prior means, standard deviations, and correlation of the state variables at time 0. The EM algorithm finds the MLE of the marginal likelihood by iteratively applying two steps:

\begin{itemize}
    \item[] \noindent\textbf{E step.} Compute conditional likelihood $p\left(\mathcal{S}_T \vert \mathcal{I}_{T};\theta_{t}\right)$.

    \item[] \noindent\textbf{M step.} Compute $\theta_{t+1}=\argmax_{\theta}\mathbb{E}\left[
\left.  \log p\left(\mathcal{S}_T, \mathcal{I}_{T};\theta\right)\right\vert
 p\left(\mathcal{S}_T \vert \mathcal{I}_{T};\theta_{t}\right)\right]$.
\end{itemize}

In particular we use the Kalman smooth to compute the likelihood in the E step and then compute optimal solution of the M step where we have a convex objective function (Appendix~\ref{app:empirical} reports a detailed description of the algorithm). The result is the following MLE of our model parameters:

\begin{table}[h]
\centering
\begin{tabular}{lr}
\toprule
\textbf{Parameter} & \textbf{Value} \\
\midrule
\(\mu^d\) & \(7.30 \times 10^7\) \\
\(\mu^{\beta}\) & \(-2.38 \times 10^{-3}\) \\
\(\alpha^d\) & \(9.97 \times 10^{-1}\) \\
\(\alpha^{\beta}\) & \(9.98 \times 10^{-1}\) \\
\(\sigma_{\epsilon^d}\) & \(1.04 \times 10^6\) \\
\(\sigma_{\epsilon^{\beta}}\) & \(5.35 \times 10^{-5}\) \\
\(\rho_{\epsilon^d, \epsilon^{\beta}}\)\phantom{000} & \(-5.48 \times 10^{-1}\) \\
\(\sigma_w\) & \(7.21 \times 10^{6}\) \\
\(d_0\) & \(4.47 \times 10^7\) \\
\(\beta_0\) & \(-9.62 \times 10^{-4}\) \\
\(\sigma_{d_0}\) & \(1.47 \times 10^6\) \\
\(\sigma_{\beta_0}\) & \(3.05 \times 10^{-4}\) \\
\(\rho_{d_0, \beta_0}\) & \(-2.57 \times 10^{-1}\) \\
\bottomrule
\end{tabular}\\[7pt]
\caption{MLE Estimates of Model Parameters via EM Algorithm with Kalman smoother.}
\label{tab:model_parameters}
\end{table}

\medskip
\noindent\textbf{\sffamily Simulation.} We now use our model to simulate the optimal policy and then compare to the EIP-1559 policy for our main results. The main outcome of our simulation is the state estimates that, together with the model parameters, fully characterize our optimal policy. We use the Kalman filter to update state estimates online and then compute the optimal policies as characterized in Section \ref{sec-optimal-policy}. The following figure reports the state estimates. We can see how the state estimates vary during the day, adapting to the patterns of observed demand. The mean for the demand sensitivity is $\bar{\beta}=-0.0025$, we can use the prevailing price to compute the implied demand elasticity $\eta_k = \beta_k p_k / t$ at the target for every block. The mean implied elasticity over the entire sample is $\bar{\eta}\approx-4$.

\begin{figure}[h]
\centering
\begin{tabular}{cc}
\begin{subfigure}{0.45\textwidth}
\includegraphics[width=\textwidth]{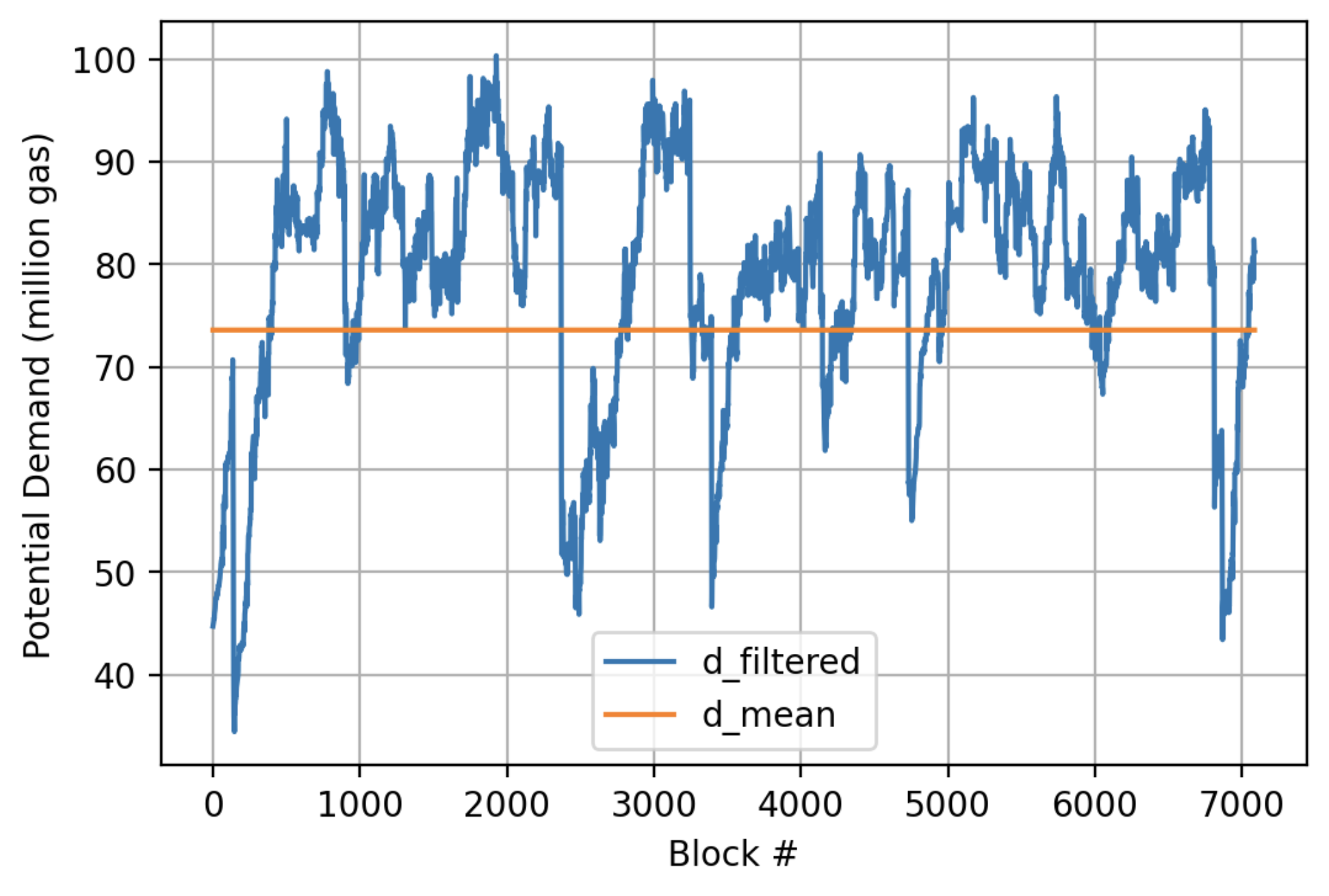}
\end{subfigure} &
\begin{subfigure}{0.475\textwidth}
\includegraphics[width=\textwidth]{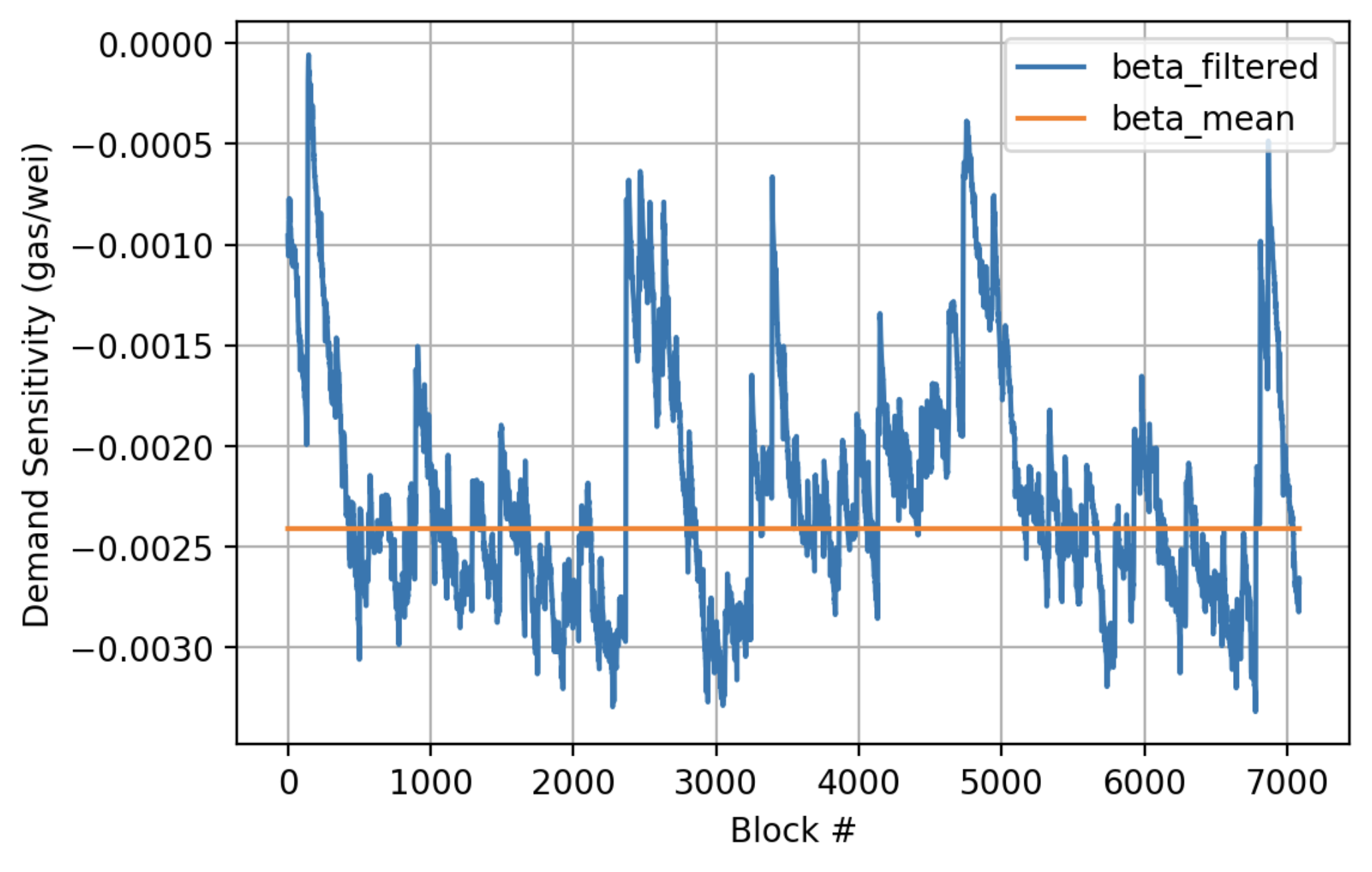}
\end{subfigure}
\end{tabular}
\caption{State estimates using the Kalman filter, potential demand $\hat{d}_k$ (left) and demand sensitivity $\hat{\beta}_k$ (right), for the first 7,080 blocks in our sample. The orange line reports the mean for the respective parameter estimate computed over \textit{all blocks in our sample}.\label{fig-state-estimates}}
\end{figure}

\subsection{Empirical Results}

We compute performance metrics for three policies and summarize them in the following table. The first policy is the optimal policy for the problem without price regularization ($\lambda=0$) that we charachterized in Theorem \ref{thm-lambda-0}, we call this policy LINDY($0$). The second policy is the policy with price regularization ($\lambda>0$), computed according to the MPC method that we derived from Theorem \ref{thm-lambda-positive}, we call this policy LINDY($\lambda$) with $\lambda=10^{-7}$ in this case. The last row reports the benchmark EIP-1559 policy. The main results are reported in Table \ref{tab-performance}, for each policy we compute the following metrics:

\begin{itemize}
    \item[$\bullet$] \textbf{Gas Used Bias: $\mathbf{Bias}(y,t)$.} Average deviation of gas used from target $t$ over all blocks in our sample. Unit: gas.
    \item[$\bullet$] \textbf{Gas Used Standard Deviation: $\mathbf{SD}(y)$.} Sample standard deviation of gas used. Unit: gas.
    \item[$\bullet$] \textbf{Root Mean Squared Deviation (of gas used from target): $\mathbf{RMSD}(y,t)$.} Square root of sum of squared difference between gas used and 15M gas target over all blocks. It can also be expressed as function of bias and standard deviation as follows
        $$\mathbf{RMSD}(y,t)=\sqrt{\frac{1}{T}\sum_{k=1}^T(y_k-t)^2}=\sqrt{\mathbf{Bias}(y,t)^2+\mathbf{SD}(y)^2}.$$
    Unit: gas.
    \item[$\bullet$] \textbf{Fraction of near-full blocks: $\mathbf{\phi_{0.95}}(y)$.} Fraction of blocks that are more than 95\% full, i.e. blocks that use more than 28.5M gas. Unit: percent.
    \item[$\bullet$] \textbf{Root Mean Squared Update: $\mathbf{RMSU}(p)$.} Square root of sum of
      squared updates, i.e. the difference between $p_{k+1}$ and $p_{k}$, for every block $k$ in
      our sample, i.e.,
$$\mathbf{RMSU}(p)=\sqrt{\frac{1}{T}\sum_{k=1}^T(p_{k}-p_{k-1})^2}.$$
      Unit: gweis.
\end{itemize}

\begin{table}[h]
\centering
\begin{tabular}{lP{3cm}P{3cm}P{3cm}P{3cm}P{3cm}}
\toprule
 & $\mathbf{Bias}(y,t)$ & $\mathbf{SD}(y)$ & $\mathbf{RMSD}(y,t)$ & $\mathbf{\phi_{0.95}}(y)$ & $\mathbf{RMSU}(p)$ \\
\addlinespace[0.5ex] 
\toprule
\textbf{LINDY($0$)} & 102 & 5,404,967 & 5,404,967 & 2.0\% & 1.9 \\
 & {\footnotesize (99.9\%)} & {\footnotesize (5.9\%)} & {\footnotesize (6.0\%)} & {\footnotesize (62\%)} & {\footnotesize (-35.7\%)} \\
\addlinespace[1ex] 
\textbf{LINDY($\lambda$)} & 3,932 & 5,404,474 & 5,404,476 & 2.0\% & 1.1 \\
 & {\footnotesize (97.2\%)} & {\footnotesize (5.9\%)} & {\footnotesize (6.0\%)} & {\footnotesize (62\%)} & {\footnotesize (21.4\%)} \\
\addlinespace[0.5ex] 
\midrule
\textbf{EIP-1559} & 140,495 & 5,745,604 & 5,747,321 & 5.3\% & 1.4 \\
\bottomrule
\end{tabular}\\[7pt]
\caption{Performance metrics for the LINDY policies and EIP-1559, the regularized policy has $\lambda=10^{-7}$. In parethesis below each metric the percentage improvement of the respective LINDY policy versus the EIP-1559 benchmark. The unit of the first three columns is Ethereum \textit{gas} and the unit of the last column is \textit{gweis}.}
\label{tab-performance}
\end{table}

The first column shows that the EIP-1559 policy overshoots the fifteen million gas target by 1\% on average in our sample while the LINDY policies are both extremely accurate, overshooting by fraction of a percent. Using column two, we can compute the \textit{RMSD} which shows that while the average deviation for the EIP-1559 policy is about 5.75M the LINDY policies deviate on average by about 5.4M, a \textit{6\% improvement}. 

Moreover, EIP-1559 results in \textit{2.65 times more near-full blocks compared to the LINDY policies}, suggesting that the performance improvement at times of high demand might be higher than the average. Finally, note that adding a bit of regularization yields a policy, LINDY($\lambda$), that is much smoother on average than EIP-1559 with an average update of only 1.1 gweis, a \textit{28\% improvement in the RMSU}. In what follows we analyze these properties with additional analyses on the performance of LINDY($\lambda$) versus the benchmark.

Figure \ref{fig-hist} reports two histograms to visualize the observed and simulated values for gas used (left) and fee updates (right) for all 99,547 blocks in our sample, comparing the EIP-1559 with the LINDY($\lambda$) policy. Note that we plot the fee update $u_k=p_{k+1}-p_k$ instead of the absolute fee. The plot on the left shows how our policy is able to achieve a distribution of gas used that is much less concentrated around the limit of 30M gas and closer to the gas target in general. At the same time, our policy shows lower variability in fees in aggregate, the distribution of updates is concentrated around the value of 0 in the right plot.

\begin{figure}[h]
\centering
\begin{tabular}{cc}
\begin{subfigure}{0.45\textwidth}
\includegraphics[width=\textwidth]{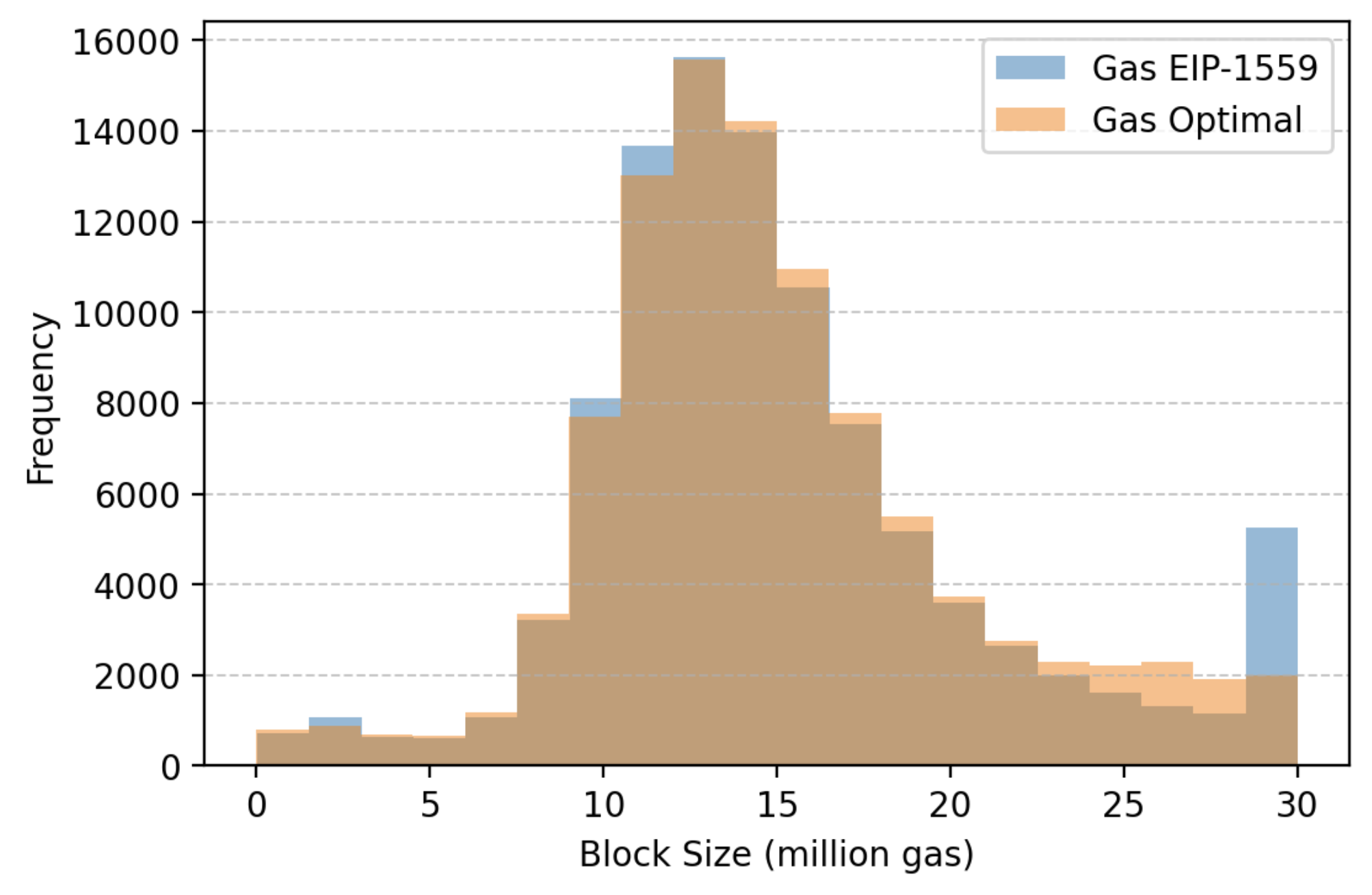}
\end{subfigure} &
\begin{subfigure}{0.45\textwidth}
\includegraphics[width=\textwidth]{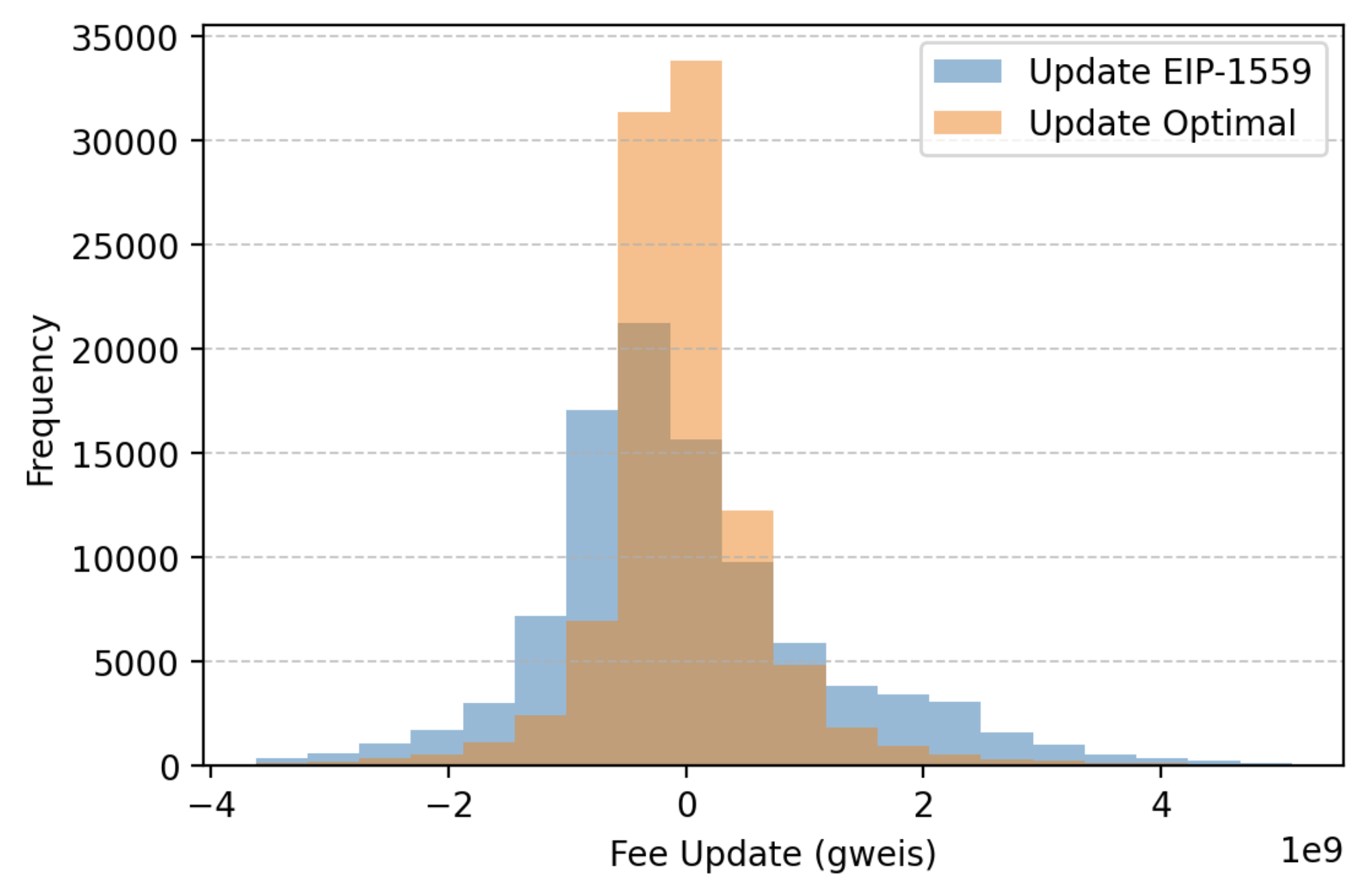}
\end{subfigure}
\end{tabular}
\caption{Histogram of gas used (left) and update size (right); comparing LINDY($\lambda$) policy to EIP-1559.\label{fig-hist}}
\end{figure}

Finally, we look at how our policy performs in different market regimes. In particular, we identify intervals of blocks that are in a \textit{demand spike} and blocks that are in a \textit{stable demand} regime. We use a heuristic that is solely based on the observed gas used. We compute a symmetric moving average of 25 blocks (roughly 5 minutes) for every block in our sample, then we classify as demand spikes the blocks that have a moving average of more than 20 million gas used (or 2/3 of the block limit) and stable demand the blocks that have a moving average between 13.5 and 16.5 million gas used (or within 1/10 of the block target).

Figure \ref{fig-spike-stable} shows an example of a demand spike and an example of stable and slightly increasing demand. We can appreciate how the LINDY($\lambda$) policy is more responsive to demand spikes, with fees surging more quickly than EIP-1559 and also going down more quickly after the demand spike (left). Also, computing the RMSD over demand spike blocks only shows that the optimal policy performs significantly better than EIP-1559, with an improvement of almost 12\%. At the same time, our policy is less variable and less reactive to noise during the many periods of stable demand (right).

\begin{figure}[h]
\centering
\begin{tabular}{cc}
\begin{subfigure}{0.45\textwidth}
\includegraphics[width=\textwidth]{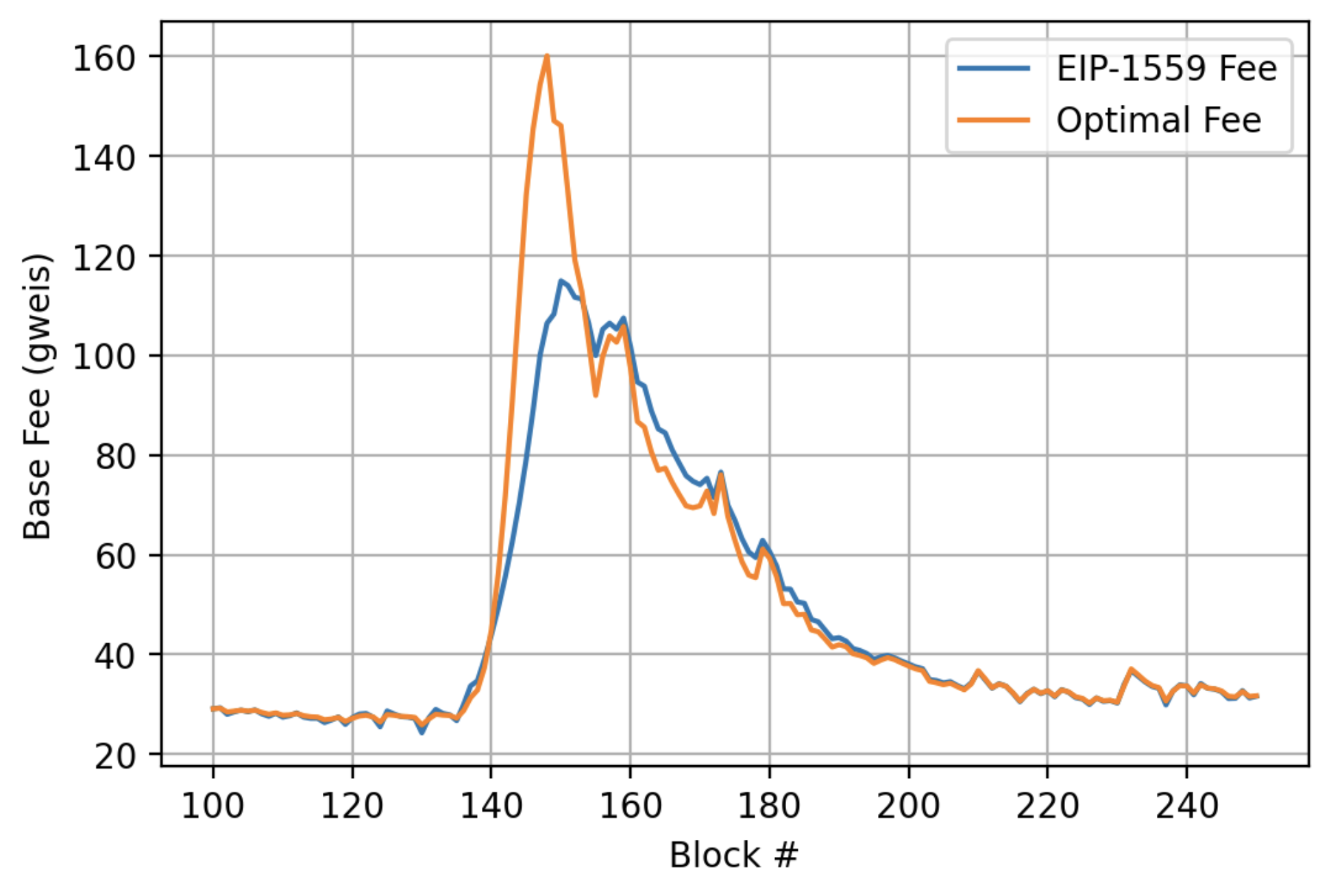}
\end{subfigure} &
\begin{subfigure}{0.442\textwidth}
\includegraphics[width=\textwidth]{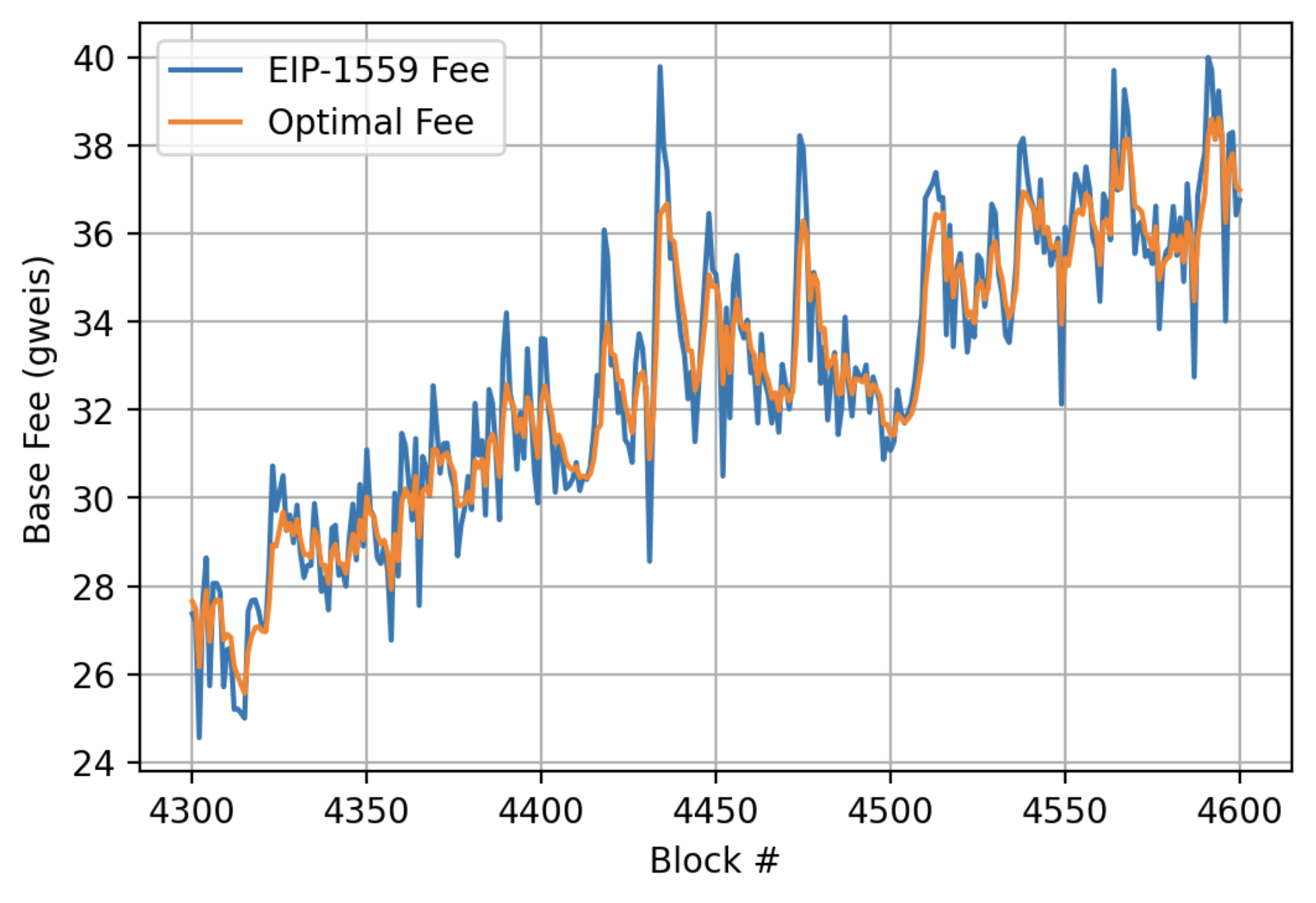}
\end{subfigure}
\end{tabular}
\caption{Base fee for EIP-1559 vs LINDY($\lambda$) policy during a demand spike (left, blocks 120 to 200) and a stable period with increasing demand (right, blocks 4300 to 4600).\label{fig-spike-stable}}
\end{figure}


\bibliographystyle{plainnat}
\bibliography{references.bib}

\newpage
\appendix

\section{Proofs}\label{app:proofs}

\noindent\textbf{\sffamily Proof of Theorem \ref{thm-lambda-0}.} We consider the objective function as a finite sum

\begin{align*}
    J_K &\defeq \mathbf{E} \left\{\sum_{k=1}^{K}  \| y_{k} - t \|_{2}^{2}\right\} \\
      &= \sum_{k=1}^{K} \mathbf{E} \left\{ \mathbf{E} \left[ \| d_{k} + B_{k} p_k + \epsilon_k^y - t \|_{2}^{2} | I_k \right]\right\}
\end{align*}
to calculate the optimal policy $p_{k+1,K}^*$.
Note that the objective $J_K$ separates into a sum of quadratic terms. Equating the first-order conditions
for $p_{k+1}$ to $0$ for every $k$ yields the control policy
$$p_{k+1,K}^* = \left[\mathbf{E}\left(B_{k+1}^\intercal B_{k+1}|I_k\right)\right]^{-1}\mathbf{E}\left(B_{k+1}^\intercal (t-d_{k+1})|I_k\right).$$
Then we let $K$ goes to infinity, and we can get
\begin{align*}
    p_{k+1}^* &= \lim_{K \to \infty} p_{k+1,K}^* \\
      &= \left[\mathbf{E}\left(B_{k+1}^\intercal B_{k+1}|I_k\right)\right]^{-1}\mathbf{E}\left(B_{k+1}^\intercal (t-d_{k+1})|I_k\right).
\end{align*}
\qed

\medskip
\noindent\textbf{\sffamily Proof of Theorem \ref{thm-lambda-positive}.} We consider the objective function as a finite sum
$$J_K = \min_{p} \mathbf{E} \left\{\sum_{k=1}^{K} \left( \| y_{k} - t \|_{2}^{2} +\lambda \| p_{k} - p_{k-1} \|_{2}^{2} \right) \right\}$$
to calculate the optimal policy $p_{k+1,K}^*$.
Using the techniques in dynamic programming, we further suppose
$$J_{k,K} \defeq \min_p \mathbb{E}\left\{
\left.  \sum_{s=k+1}^{K}\left(  \left\Vert y_{s}-t\right\Vert
_{2}^{2}+\lambda\left\Vert p_s-p_{s-1} \right\Vert _{2}^{2}\right)  \right\vert
I_{k}\right\}.$$
\begin{lemma}
    We can use induction to show that there exists matrices $Q_{k,K}, R_{k,K} \in\mathbb{R}^{n \times n}$ and vector $\tau_{k,K}\in\mathbb{R}^{n}$ such that
    $$J_{k,K} = p_{k}^\intercal Q_{k,K}p_{k}-2\tau_{k,K}^\intercal p_{k}+2(a_{k+1}^d)^\intercal R_{k,K} p_{k}+H_{k,K},$$
    where $a_{k+1}^d \defeq \mathbb{E}\left(d_{k+1}|I_k \right)$ and $H_{k,K}$ is independent with the choice of decision variables $p$.
    Moreover, $Q_{k,K}, R_{k,K}, \tau_{k,K}$ satisfy
    $$Q_{k-1,K} = \lambda I -\lambda^2(\lambda I + Q_{k,K} + B_k^\intercal B_k)^{-1},$$
    $$R_{k-1,K}^\intercal = \lambda (\lambda I + Q_{k,K} + B_k^\intercal B_k)^{-1} (B_k^\intercal + R_{k,K}^\intercal A^d),$$
    $$\tau_{k-1,K} = \lambda (\lambda I + Q_{k,K} + B_k^\intercal B_k)^{-1} \left(B_k^\intercal t + \tau_{k,K} - R_{k,K}^\intercal(I-A^d)\mu^d\right).$$
\label{lemma-expression-J}
\end{lemma}
\begin{proof}
    For $k=K$, $J_{K,K} = 0$, we can let $Q_{K,K}=0$, $R_{K,K}=0$, and $\tau_{K,K}=0$.
    Suppose for $k = s$, we have
    $$J_{s,K} = p_{s}^\intercal Q_{s,K}p_{s}-2\tau_{s,K}^\intercal p_{s}+2(a_{s+1}^d)^\intercal R_{s,K} p_{s}+H_{s,K}.$$
    For $k=s-1$, we can first get the optimal policy
    \begin{align*}
    p_{s,K}^* &= \argmin_{p_{s}} \mathbf{E} \left\{ \| y_{s} - t \|_{2}^{2} +\lambda \| p_{s} - p_{s-1} \|_{2}^{2} + J_{s,K} | I_{s-1} \right\} \\
      &= \argmin_{p_{s}} \mathbf{E} \left\{ \| d_{s} + B_s p_s +\epsilon_s^y - t \|_{2}^{2} +\lambda \| p_{s} - p_{s-1} \|_{2}^{2} + p_{s}^\intercal Q_{s,K}p_{s}-2\tau_{s,K}^\intercal p_{s}+2(a_{s+1}^d)^\intercal R_{s,K} p_{s}| I_{s-1} \right\},
    \end{align*}
    by the expression of $J_{s,K}$. It is a quadratic function with respect to the decision variable $p_{s}$. We can get the optimal policy
    $$p_{s,K}^* = -(\lambda I+Q_{s,K}+B_s^\intercal B_s)^{-1}\left(B_s^\intercal a_s^d-B_s^\intercal t-\lambda p_{s-1}-\tau_{s,K}+R_{s,K}^\intercal\left[(I-A^d)\mu^d+A^d a_s^d\right]\right),$$
    by equating its first-order condition to $0$.
    After plugging in and comparing the coefficient, the optimal value
    $$J_{s-1,K} = p_{s-1}^\intercal Q_{s-1,K}p_{s-1}-2\tau_{s-1,K}^\intercal p_{s-1}+2(a_{s}^d)^\intercal R_{s-1,K} p_{s-1}+H_{s-1,K},$$
    where
    $$Q_{s-1,K} = \lambda I -\lambda^2(\lambda I + Q_{s,K} + B_s^\intercal B_s)^{-1},$$
    $$R_{s-1,K}^\intercal = \lambda (\lambda I + Q_{s,K} + B_s^\intercal B_s)^{-1} (B_s^\intercal + R_{s,K}^\intercal A^d),$$
    $$\tau_{s-1,K} = \lambda (\lambda I + Q_{s,K} + B_s^\intercal B_s)^{-1} \left(B_s^\intercal t + \tau_{s,K} - R_{s,K}^\intercal(I-A^d)\mu^d\right).$$
    \qed
\end{proof}
Suppose $Q_k = \lim_{K \to \infty} Q_k^K$, $R_k = \lim_{K \to \infty} R_k^K$, $\tau_k = \lim_{K \to \infty} \tau_k^K$, $H_k = \lim_{K \to \infty} (H_{k,K}-H_{0,K})$, and $\Tilde{J}_k = p_{k}^\intercal Q_{k}p_{k}-2\tau_{k}^\intercal p_{k}+2(a_{k+1}^d)^\intercal R_{k} p_{k}+H_k$.
Taking limit in Lemma \ref{lemma-expression-J}, we have
$$p_{k+1}^* = \argmin_{p_{k+1}} \mathbf{E} \left\{ \| y_{k+1} - t \|_{2}^{2} +\lambda \| p_{k+1} - p_{k} \|_{2}^{2} + \Tilde{J}_{k+1} | I_{k} \right\},$$
\begin{equation}
   \Tilde{J}_{k} = \min_{p_{k+1}} \mathbf{E} \left\{ \| y_{k+1} - t \|_{2}^{2} +\lambda \| p_{k+1} - p_{k} \|_{2}^{2} + \Tilde{J}_{k+1} | I_{k} \right\}.
\label{equation-lim-J}
\end{equation}
We take the partial derivative with respect to $p_k$ on both sides of the equation \ref{equation-lim-J}. Then we have
$$2Q_kp_k - 2\tau_k + 2R_k^\intercal a_{k+1}^d = -2 \lambda (p_{k+1}^*-p_k),$$
which means that
$$p_{k+1}^* =\left( I-\frac{Q_k}{\lambda} \right) p_k + \frac{Q_k}{\lambda} \aim_k,$$
where $\aim_k = Q_k^{-1}(\tau_k-R_k^\intercal a_{k+1}^d).$

We can also get
$$Q_{k} = \lambda I -\lambda^2(\lambda I + Q_{k+1} + B_{k+1}^\intercal B_{k+1})^{-1},$$
$$R_{k}^\intercal = \lambda (\lambda I + Q_{k+1} + B_{k+1}^\intercal B_{k+1})^{-1} (B_{k+1}^\intercal + R_{k+1}^\intercal A^d),$$
$$\tau_{k} = \lambda (\lambda I + Q_{k+1} + B_{k+1}^\intercal B_{k+1})^{-1} \left(B_{k+1}^\intercal t + \tau_{k+1} - R_{k+1}^\intercal(I-A^d)\mu^d\right),$$
after taking limit. By some algebraic manipulations, we can get
$$Q_k^{-1} = \lambda^{-1}(Q_{k+1}+ B_{k+1}^\intercal B_{k+1})^{-1}(\lambda I + Q_{k+1} + B_{k+1}^\intercal B_{k+1}),$$
$$\tau_k-R_k^\intercal a_{k+1}^d = \lambda (\lambda I + Q_{k+1} + B_{k+1}^\intercal B_{k+1})^{-1} \left(
B_{k+1}^\intercal (t - a_{k+1}^d) + \tau_{k+1} - R_{k+1}^\intercal\left((I-A^d)\mu^d + A^d a_{k+1}^d\right)
\right).$$
Thus, we can get
\begin{align*}
    \aim_k &= Q_k^{-1}(\tau_k-R_k^\intercal a_{k+1}^d) \\
           &= (Q_{k+1}+ B_{k+1}^\intercal B_{k+1})^{-1} \left( B_{k+1}^\intercal (t - a_{k+1}^d) + \tau_{k+1} - R_{k+1}^\intercal\left((I-A^d)\mu^d + A^d a_{k+1}^d\right)\right) \\
           &= (B_{k+1}^\intercal B_{k+1}+Q_{k+1})^{-1} \left( B_{k+1}^\intercal B_{k+1} \Bar{p}_{k+1} + Q_{k+1} \mathbf{E}(\aim_{k+1}|I_k)\right),
\end{align*}
where $\Bar{p}_{k+1} \defeq B_{k+1}^{-1}(t - a_{k+1}^d)$ satisfying $\mathbf{E}(y_{k+1}|I_{k},p_{k+1} = \Bar{p}_{k+1})=t.$ is the market clearing price at block $k+1$. We can obtain the result stated in Theorem \ref{thm-lambda-positive} by employing the same technique to represent $\alpha_{k+1}$ and iterating this process repeatedly. \qed

\medskip
\noindent\textbf{\sffamily Proof of Theorem \ref{thm:eigen}.} By Theorem \ref{thm-lambda-0}, the optimal policy
\begin{align*}
    p_{k+1}^* &= \left[\mathbf{E}\left(B_{k+1}^\intercal B_{k+1}|I_k\right)\right]^{-1}\mathbf{E}\left(B_{k+1}^\intercal (t-d_{k+1})|I_k\right). \\
    &= \left[\mathbf{E}\left(V\diag(\delta_{k+1})U^\intercal U\diag(\delta_{k+1})V^\intercal|I_k\right)\right]^{-1}\mathbf{E}\left(V\diag(\delta_{k+1})U^\intercal (U\Tilde{t}-U\Tilde{d}_{k+1})|I_k\right)\\
    &= V \left[\mathbf{E}\left(\diag(\delta_{k+1})^2|I_k\right)\right]^{-1}\mathbf{E}\left(\diag(\delta_{k+1})(\Tilde{t}-\Tilde{d}_{k+1})|I_k\right).
\end{align*}
By the eigenprice $\Tilde{p}_{k+1}^* = V^\intercal p_{k+1}^*$, we can obtain
$$\Tilde{p}_{k+1}^* = \left[\mathbf{E}\left(\diag(\delta_{k+1})^2|I_k\right)\right]^{-1}\mathbf{E}\left(\diag(\delta_{k+1})(\Tilde{t}-\Tilde{d}_{k+1})|I_k\right),$$
which is equivalent to Theroem \ref{thm:eigen}. \qed


\section{Empirical Models}\label{app:empirical}

We introduce the Expectation-Maximization (EM) algorithm to estimate the parameters in our model using real data. The EM algorithm is a two-step iterative method designed for finding maximum likelihood estimates of parameters in probabilistic models with incomplete data. It is especially useful in situations where direct optimization of the likelihood is computationally challenging or not feasible.

In our study, the EM algorithm is applicable to both our baseline model and eigenresources model. For clarity and brevity, we will take the baseline model as an illustrative example to introduce how the EM algorithm works.

\medskip
\noindent\textbf{\sffamily Model review.} Demand dynamics
$$d_{k+1} = (I-A^{d})\mu^{d} + A^{d} d_{k} + \epsilon_k^{d}.$$
Price-sensitivity matrix
$$\mathbf{vec}(B_{k+1}) = (I-A^{B})\mu^{B}+A^{B}\mathbf{vec}({B_k})+\epsilon_k^{B}.$$
Observed demand
$$y_{k} = d_{k} + B_{k} p_{k} + \epsilon_{k}^{y}.$$

\medskip
\noindent\textbf{\sffamily Hidden state and observed state.} The hidden state in our model is $x_k^\intercal = \left[ d_k^\intercal \  \mathbf{vec}(B_k)^\intercal \right].$ It evolves according to
$$x_{k+1}=(\mathrm{I}-A^{x})\mu^{x}+A^{x}x_k+\epsilon_k^{x}, \quad\quad \epsilon_k^{x} \sim \mathcal{N}(0, W^{x}).$$
The observed state in our model is $y_k$. It satisfies
$$y_k = C_k x_k + \epsilon_k^y, \quad\quad \epsilon_k^{y} \sim \mathcal{N}(0, W^{y}).$$

\medskip
\noindent\textbf{\sffamily Parameters.} The parameter that need to be estimated in our model is
$$\theta \defeq (\mu^x, A^x, W^x, W^y, a_0, S_0),$$
where $x_0$ and $S_0$ represent the normal-mean and normal-variance for the initial hidden state respectively.

\medskip
\noindent\textbf{\sffamily EM algorithm.}
\begin{enumerate}
    \item \textbf{Initialization:} Choose initial values $\theta_0$ for the parameters.
    \item \textbf{E-Step (Expectation):} Given the current parameters $\theta_t$, compute the expectation of the log likelihood 
    $$l(\theta|\theta_t) \defeq \mathbf{E}\left\{log(x_{0:T},y_{0:T};\theta)|x_{0:T},y_{0:T} \sim p(\theta_t), I_T\right\}$$
    with respect to the hidden state and observed state.
    \item \textbf{M-Step (Maximization):} Maximize the expected log likelihood found in the E-step with respect to the parameter $\theta$, and update the parameter to this new value, i.e.,
    $\theta_{t+1} = \argmax_{\theta} l(\theta|\theta_t).$
\end{enumerate}
\noindent Repeat the E-Step and M-Step until the algorithm converges.


\end{document}